\documentclass{article}
\usepackage[a4paper]{geometry}
\usepackage{amssymb,amsmath,amsthm}
\usepackage{graphics,graphicx,color}
\usepackage{microtype,verbatim,dsfont}
\usepackage{xspace,enumerate}
\usepackage{subcaption,paralist}
\usepackage{todonotes}
\usepackage[pdfpagelabels,colorlinks,citecolor=blue,linkcolor=blue,urlcolor=blue]{hyperref}
\usepackage[capitalise]{cleveref}

\theoremstyle{plain}
\newtheorem{theorem}{Theorem}

\newtheorem{lemma}[theorem]{Lemma}

\newtheorem{observation}[theorem]{Observation}
\newtheorem{property}[theorem]{Property}

\Crefname{observation}{Observation}{Observations}
\Crefname{algorithm}{Algorithm}{Algorithms}
\Crefname{algocf}{Algorithm}{Algorithms}
\Crefname{section}{Section}{Sections}
\Crefname{lemma}{Lemma}{Lemmata}
\Crefname{claim}{Claim}{Claims}
\Crefname{property}{Property}{Properties}
\Crefname{enumi}{Property}{Properties}
\Crefname{figure}{Fig.}{Figs.}

\newcommand{\myparagraph}[1]{\medskip\noindent\textbf{#1}}

\newcommand{\drawing}{disk-link drawing\xspace}
\newcommand{\drawings}{\drawing{s}\xspace}

\title{Convex Grid Drawings of Planar Graphs\\with Constant Edge-Vertex Resolution}

\author{Michael A. Bekos$^1$, Martin~Gronemann$^2$,\\ Fabrizio~Montecchiani$^3$, Antonios~Symvonis$^4$~\vspace{8pt}\\
\small{$^1$Department of Mathematics, University of Ioannina, Ioannina, Greece}\\
\small\texttt{bekos@uoi.gr}
\\
\small{$^2$Algorithms and Complexity Group, TU Wien, Vienna, Austria}\\
\small\texttt{mgronemann@ac.tuwien.ac.at}
\\
\small{$^3$Department of Engineering, University of Perugia, Perugia, Italy}\\
\small\texttt{fabrizio.montecchiani@unipg.it}
\\
\small{$^4$School of Applied Mathematical \& Physical Sciences},\\\small{National Technical University of Athens, Athens, Greece}\\
\small\texttt{symvonis@math.ntua.gr}
}
\date{}

\begin{document}

\maketitle 

\begin{abstract}
We continue the study of the area requirement of convex straight-line grid drawings of $3$-connected plane graphs, which has been intensively investigated in the last decades. Motivated by applications, such as graph editors, we additionally require the obtained drawings to have bounded \emph{edge-vertex resolution}, that is, the closest distance between a vertex and any non-incident edge is lower bounded by a constant that does not depend on the size of the graph. We present a drawing algorithm that takes as input a $3$-connected plane graph with $n$ vertices and $f$ internal faces and computes a convex straight-line drawing with edge-vertex resolution at least $\frac{1}{2}$ on an integer grid of size $(n-2+a) \times (n-2+a)$, where $a = \min\{n-3,f\}$. Our result improves the previously best-known area bound of $(3n-7) \times (3n-7)/2$ by Chrobak, Goodrich and Tamassia. 
\end{abstract}

\section{Introduction}
\label{sec:introduction}

Fáry’s theorem~\cite{Far48} is a fundamental result in planar graph drawing, as it guarantees the existence of a planar straight-line drawing for every planar graph. In such a drawing, the vertices of the graph are mapped to distinct points of the Euclidean plane in such a way that the edges are straight, non-intersecting line-segments. This central result has been independently proved by several researchers in early works~\cite{St51,SH34,Wag36}, some of which also suggested corresponding constructive algorithms requiring high-precision arithmetics; see, e.g.,~\cite{CON85,Tu63}. In this regard, a breakthrough has been introduced by de Fraysseix, Pach and Pollack~\cite{DBLP:conf/stoc/FraysseixPP88} in the late 80’s, who proposed a method that additionally guarantees the obtained drawings to be on an integer grid (thus making the high-precision operations unnecessary). A linear-time implementation of this method was proposed by Chrobak and Payne~\cite{DBLP:journals/ipl/ChrobakP95}. Over the years, several works have studied the area requirement of planar graphs under different settings, by providing bounds on the required size of the underlying grid; see, e.g.,~\cite{DBLP:journals/algorithmica/BattistaTV99,DBLP:journals/order/Felsner01,DBLP:journals/dcg/He97,DBLP:journals/dcg/MiuraNN01,DBLP:conf/soda/Schnyder90}. In the original work by de Fraysseix et al. the size of the underlying grid is $(2n-4) \times (n-2)$ with $n$ being the number of  vertices of the graph; such a bound is asymptotically worst-case optimal, as it is known that there exist $n$-vertex planar graphs that need $\Omega(n) \times \Omega(n)$ area in any of their planar drawings~\cite{DBLP:conf/stoc/FraysseixPP88,DBLP:journals/algorithmica/Kant96}.

The corresponding best-known\footnote{Note that improvements on this bound are known but they are obtained by exploiting either the structure of the input graph~\cite{DBLP:journals/algorithmica/BonichonFM07,DBLP:journals/algorithmica/BattistaTV99,DBLP:journals/order/Felsner01,DBLP:conf/wads/ZhangH03} or higher connectivity~\cite{DBLP:journals/dcg/He97,DBLP:journals/dcg/MiuraNN01}.} upper bound is due to Chrobak and Kant~\cite{DBLP:journals/ijcga/ChrobakK97}, who presented a linear-time algorithm to embed any $n$-vertex planar graph into a grid of size $(n-2) \times (n-2)$; see also~\cite{DBLP:conf/soda/Schnyder90}. In contrast to the work by de Fraysseix, Pach and Pollack~\cite{DBLP:conf/stoc/FraysseixPP88}, which requires an augmentation of the input planar graph to maximal planar, the algorithm by Chrobak and Kant~\cite{DBLP:journals/ipl/ChrobakP95} requires just $3$-connectivity. Furthermore, it guarantees an additional property, which is desired when drawing $3$-connected planar graphs (see, e.g.,~\cite{DBLP:journals/jct/Thomassen84}): the obtained drawings are \emph{convex}, i.e., the boundary of each face is a convex polygon.

\begin{figure}[t]
    \centering
    \begin{subfigure}{.3\textwidth}
    \centering
	\includegraphics[width=0.6\textwidth,page=1]{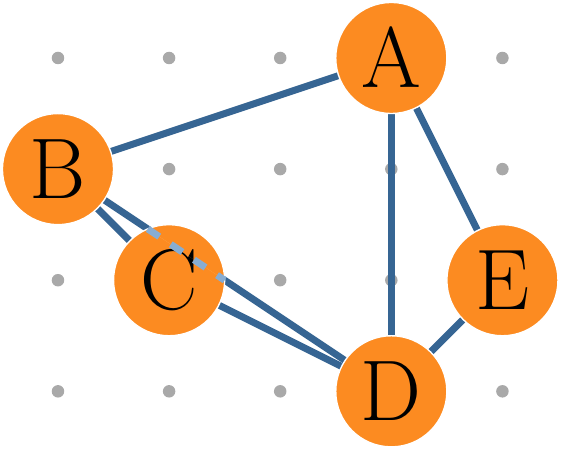}
	\caption{}
	\end{subfigure}
    \begin{subfigure}{.3\textwidth}
    \centering
    \includegraphics[width=0.6\textwidth,page=2]{figs/disk}
	\caption{}
    \end{subfigure}
    \caption{Two planar straight-line grid drawings of the same graph; the drawing in (a) contains an edge-vertex intersection (on vertex C), while the  one in (b) does not as it has edge-vertex resolution at least $\frac{1}{2}$.}
    \label{fig:intro}
\end{figure}

Back in 1996, Chrobak, Goodrich and Tamassia~\cite{DBLP:conf/compgeom/ChrobakGT96} studied the area requirement of $3$-connected planar graphs under an additional requirement, which is essential in practical applications. In particular, they introduced the notion of \emph{edge-vertex resolution}, which measures how close a vertex is to any non-incident edge, and required that the obtained drawings have bounded edge-vertex resolution. This requirement becomes essential in several practical situations, for instance, consider graph editors which usually represent each vertex by an object of a certain size (rather than a point) containing a distinguishing label. Having high edge-vertex resolution allows to avoid  potential overlaps between vertices and edges, in particular, having edge-vertex resolution at least $\frac{1}{2}$ allows each vertex to be represented as an open disk of unit diameter, such that overlaps between vertices and non-incident edges are completely avoided, and simultaneously vertices centered at neighboring grid-points do not overlap (although may touch); see \cref{fig:intro}. In their work~\cite{DBLP:conf/compgeom/ChrobakGT96}, Chrobak, Goodrich and Tamassia claimed that every $3$-connected planar graph admits a convex planar straight-line grid drawing on a grid of size $(3n-7)\times (3n-7)/2$ with edge-vertex resolution at least $\frac{1}{2}$. However, the details of the algorithm (and of its proof) supporting this claim never appeared in the literature. In this regard, very recently, Bekos et al.~\cite{DBLP:journals/comgeo/BekosGMPST21} referred to the drawings with edge-vertex resolution at least $\frac{1}{2}$ as  \emph{disk-link} and proved (among other results) that every  planar graph admits a planar straight-line disk-link drawing on a grid of size $(3n-7) \times (3n-7)/2$. However, the obtained drawing is not necessarily convex. 

\myparagraph{Contribution.} 
We improve both results mentioned above by providing a linear-time algorithm to compute  planar straight-line \drawings  that are convex and that fit on a grid of size $(n-2+a) \times (n-2+a)$, where $a = \min\{f,n-3\}$ and $f$ denotes the number of internal faces of the input graph. In particular, if the input graph is maximal planar (that is, $f=2n-5$), our technique yields drawings of area $(2n-5)\times (2n-5)$. On the other hand, if the input graph is $3$-connected cubic (that is, $f=\frac{n}{2}+1$), then our technique yields drawings of area $(\frac{3n}{2}-1) \times (\frac{3n}{2}-1)$. Our result is summarized in the next theorem.

\begin{theorem}\label{th:main}
Every $3$-connected plane graph with $n$ vertices and $f$ internal faces admits a convex planar straight-line grid drawing with edge-vertex resolution at least $\frac{1}{2}$ on a grid of size $(n-2+a) \times (n-2+a)$, where $a = \min\{f,n-3\}$.  Also, the drawing can be computed in $O(n)$~time.
\end{theorem}

\myparagraph{Related work.} Bárány and Rote~\cite{Barany2006} prove that every  $3$-connected planar graph has a \emph{strictly} convex drawing on a quartic grid, improving a previous result by Rote~\cite{DBLP:conf/soda/Rote05}. We recall that a planar drawing is strictly convex if each face is bounded by a strictly convex polygon.  We point the interested reader to the surveys by Di Battista and Frati~\cite{Battista2013,DBLP:journals/corr/BattistaF14} for additional references and results concerning  convex and strictly-convex drawings of planar graphs in small area.

Concerning the edge-vertex resolution requirement there exist multiple related streams of research. A \emph{closed rectangle-of-influence} (closed RI for short) drawing is a planar straight-line drawing such that no vertex lies in the axis-parallel rectangle (including the boundary) defined by the two ends of every edge~\cite{DBLP:conf/gd/AlamdariB12,DBLP:journals/dm/BarriereH12,DBLP:conf/gd/BiedlBM99,DBLP:conf/cccg/BiedlLMV16,DBLP:journals/dcg/MiuraMN09,DBLP:journals/comgeo/SadasivamZ11}. Any closed RI drawing whose vertices are at integer coordinates can be seen as a \drawing. This implies that \drawings (not necessarily convex) in quadratic area exist for several classes of plane graphs~\cite{DBLP:journals/dm/BarriereH12,DBLP:conf/gd/BiedlBM99,DBLP:journals/comgeo/SadasivamZ11}. However, any plane graph with a filled $3$-cycle does not admit a closed RI drawing~\cite{DBLP:conf/gd/BiedlBM99}. Another related direction considers drawings where vertices are objects with integer coordinates and the edges are fat segments~\cite{DBLP:journals/jgaa/BarequetGR04}. In such drawings the edges do not connect the centers of the incident vertex-disks but rather simply enter these vertex-objects through varying angles. Duncan et al.~\cite{DBLP:journals/ijfcs/DuncanEKW06} also use fat edges but, in contrast to~\cite{DBLP:journals/jgaa/BarequetGR04}, they do not compute a drawing from scratch but rather try to extend an existing one without modifying the area of the layout. Van Kreveld~\cite{DBLP:journals/comgeo/Kreveld11} studies \emph{bold drawings}, in which vertices are drawn as disks of radius $r$ and edges as rectangles of width~$w$, where $r>w/2$. A bold drawing is \emph{good} if all of its vertices and edges are at least partially visible (neither a vertex disk nor an edge-rectangle is completely hidden by overlapping edges). Although \drawing{s} form a special case of bold drawings in which $r=\frac{1}{2}-\varepsilon$ and $w=2\varepsilon$ (for some sufficiently small $\varepsilon>0$), the research on bold drawings has mainly focused on finding feasible values of $r$ and $w$, rather than on area bounds for fixed values of $r$ and~$w$. 

\section{Preliminaries}
\label{sec:basics}

\noindent\textbf{Basic definitions.} A \emph{drawing} of a graph maps each vertex to a distinct point of the Euclidean plane, and each edge to a Jordan arc connecting its endpoints. A drawing of a graph is \emph{planar} if no two edges intersect, except possibly at a common endpoint. A planar drawing partitions the plane into topologically connected regions, which are commonly called \emph{faces}. The unbounded region is called \emph{outer face}; any other face is an \emph{internal face}.  A graph is \emph{planar} if and and only if it admits a planar drawing. A \emph{planar embedding} of a planar graph is an equivalence class of topologically-equivalent (i.e., isotopic) planar drawings. A planar graph with a given planar embedding is a \emph{plane graph}.

A drawing is \emph{straight-line} if the Jordan arcs representing the edges are straight-line segments. The \emph{slope} of a line $\ell$ is the tangent of the minimum-angle that a horizontal line needs to be rotated in order to make it overlap with $\ell$; a positive slope corresponds to a counter-clockwise rotation, while a negative one corresponds to a clockwise rotation. The \emph{slope} of a segment is the slope of the supporting line containing it. A \emph{grid drawing} of a graph is a straight-line drawing whose vertices are at integer coordinates.  We say that the grid size of a grid drawing $\Gamma$ is $W \times H$ (or, equivalently, the area of $\Gamma$ is $W \times H$), if the minimum axis-aligned box containing $\Gamma$ has side lengths $W-1$ and $H-1$. Moreover, for a vertex $v$ of a graph $G$, we denote by $x_\Gamma(v)$ and by $y_\Gamma(v)$ the $x$- and $y$-coordinate of $v$ in drawing $\Gamma$ of $G$, respectively. When the reference to $\Gamma$ is clear from the context, we simply write $x(v)$ and  $y(v)$.

\medskip\noindent\textbf{Disk-link drawings.} The \emph{edge-vertex resolution} of a grid drawing of a graph is the minimum Euclidean distance between a point representing a vertex and any edge that is not incident to that vertex. A \emph{\drawing} of a graph is a grid drawing of edge-vertex resolution at least $\frac{1}{2}$. Observe that, in a \drawing~$\Gamma$, for each vertex $v$ one can draw an open disk with radius~$\overline{\rho}\le \frac{1}{2}$ centered at the point of $\Gamma$ representing $v$, and this results in a diagram in which no two disks intersect, and no disk is intersected by a non-incident edge. For simplicity, we assume that $\overline{\rho}=\frac{1}{2}$, i.e., the disks have unit diameter. This assumption is not restrictive, since our results carry over for any constant radius up to some multiplicative constant factor for the area. 

\medskip\noindent\textbf{Canonical order.} Even though we assume familiarity with basic concepts of planar graph drawing~\cite{DBLP:books/ws/NishizekiR04,DBLP:reference/crc/Vismara13}, we recall in this section a key concept that is central in several algorithms for producing planar grid drawings of plane graphs, e.g.,~\cite{DBLP:conf/compgeom/ChrobakGT96,DBLP:conf/stoc/FraysseixPP88,DBLP:journals/algorithmica/Kant96}. Namely, the \emph{canonical order}~\cite{DBLP:journals/algorithmica/Kant96} for 3-connected plane graphs, which is defined as follows: Let $G$ be a $3$-connected plane graph with $n$ vertices and let $\pi = (P_0,\ldots,P_m)$ be a partition of the vertex-set of $G$ into paths, such that $P_0 = \{v_1,v_2\}$, $P_m=\{v_n\}$, and edges $(v_1,v_2)$ and $(v_1,v_n)$ exist and belong to the outer face of $G$. For $k=0,\ldots,m$, let $G_k$ be the subgraph induced by $\cup_{i=0}^k P_i$ and denote by $C_k$ the \emph{contour} of $G_k$ defined as follows: If $k=0$, then $C_0$ is the edge $(v_1,v_2)$, while if $k>0$, then $C_k$ is the path from $v_1$ to $v_2$ obtained by  removing $(v_1,v_2)$ from the cycle delimiting the outer face of $G_k$. We say that $\pi$ is a \emph{canonical order} of $G$ if for each $k=1,\ldots,m-1$ the following properties hold: %
\begin{enumerate}[P.1]
\item\label{prp:1}$G_k$ is biconnected and internally $3$-connected,
\item\label{prp:2}all neighbors of $P_k$ in $G_{k-1}$ are on $C_{k-1}$,
\item\label{prp:3}either $P_k$ is a \emph{singleton} (that is, $|P_k|=1$), or $P_k$ is a \emph{chain} (that is, $|P_k|>1$) and the degree of each vertex of $P_k$ is $2$ in $G_k$, and
\item\label{prp:4}all vertices of $P_k$ with $0\leq k < m$ have at least one neighbor in $P_j$ for some $j > k$.
\end{enumerate}
A canonical order of $G$ can be computed in linear time~\cite{DBLP:journals/algorithmica/Kant96}.
A vertex on contour $C_k$ is called \emph{saturated} in $G_k$ if and only if it is not adjacent to a vertex belonging to a path $P_{k'}$ with $k' > k$. 

\section{Convex planar grid \drawings}\label{sec:main}

In this section, we present our algorithm to compute convex planar grid \drawings of 3-connected plane graphs. As our algorithm builds upon an algorithm by Chrobak and Kant~\cite{DBLP:journals/ijcga/ChrobakK97}, which yields convex planar grid drawings (that are not necessarily disk-link) of 3-connected plane graphs with $n$ vertices on grids of size $(n-2) \times (n-2)$, for completeness, we first recall its basic ingredients before we enter the details of our approach.

\subsection{The algorithm by Chrobak and Kant~\cite{DBLP:journals/ijcga/ChrobakK97}.}

This algorithm is incrementally computing a convex planar drawing $\Gamma$ of a 3-connected plane graph $G$ using a canonical order $\pi = (P_0,\ldots,P_m)$ of $G$. The drawing $\Gamma$ has integer grid coordinates and fits in a grid of size $(n-2)\times (n-2)$. In order to ease the presentation, we define a Schnyder-like~\cite{felsner,DBLP:conf/soda/Schnyder90} $4$-coloring of the edges of $G$ based on the canonical order $\pi$. $G_0$ consists of a single edge $(v_1,v_2)$, which is assigned the black color. Assuming that a $4$-coloring has been constructed for $G_{k-1}$ with $k=1,\ldots,m$, we extend it for $G_k$ as follows (see \cref{fig:contour-condition}): We first color the edges of $G_{k}$ that do not belong to $G_{k-1}$ and are on contour $C_{k}$. We color the first such edge encountered in a traversal of $C_k$ from $v_1$ to $v_2$ blue, the last one green and all remaining ones (i.e., those having both endpoints in $P_k$, when $P_k$ is a chain) black. Similar to the Schnyder coloring of maximal planar graphs, we assign the color red to the remaining edges of $G_k$ that do not belong to $G_{k-1}$ (i.e., those that are incident to $P_k$ and are not part of contour $C_k$). Note that the latter case only arises if $P_k$ is a singleton by Property~P.\ref{prp:3} of the canonical~order. 

\begin{figure}[t!]
	\centering
	\begin{subfigure}{.48\textwidth}
	   \centering
	\includegraphics[page=7]{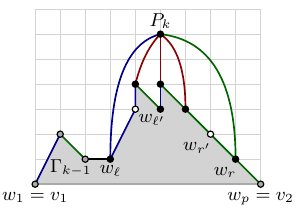}
	\subcaption{Contour condition}
	\label{fig:contour-condition}
	\end{subfigure}
	\hfil
	\begin{subfigure}{.48\textwidth}
	\centering
	\includegraphics[page=8]{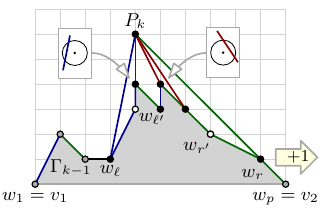}
	\subcaption{Placement of $P_k$ in $\Gamma_{k-1}$}
	\label{fig:original}
	\end{subfigure}
	\caption{Introducing a singleton $P_k$ in $\Gamma_{k-1}$ according to the algorithm by Chrobak and Kant~\cite{DBLP:journals/ijcga/ChrobakK97}; the white-filled vertices are the critical vertices $w_{\ell'}$ and $w_{r'}$.}
	\label{fig:chrobak-kant}
\end{figure}

Based on the canonical order $\pi$ of $G$, drawing $\Gamma$ is constructed as follows: Initially, the vertices $v_1$ and $v_2$ of $P_0$ are placed at points $(0,0)$ and $(1,0)$, respectively. For $k=1,\ldots,m$, assume that a planar convex grid drawing $\Gamma_{k-1}$ of $G_{k-1}$ has been constructed in which the edges of contour $C_{k-1}$ are drawn as straight-line segments with slopes $0$, $-1$ or in $[1,+\infty]$ (\emph{contour condition}; see \cref{fig:contour-condition}). In particular, the slope of each blue edge of $C_{k-1}$ is at least $1$, the slope of each black edge of $C_{k-1}$ is $0$, while the slope of each green edge of $C_{k-1}$ is $-1$ (note that $C_{k-1}$ does not contain any red edge by definition). Also, each vertex $v$ in $G_{k-1}$ has been associated with a so-called \emph{shift-set}, denoted by $S(v)$; the shift-sets of $v_1$ and $v_2$ of path $P_0$ are singletons such that $S(v_1)=\{v_1\}$ and $S(v_2)=\{v_2\}$.

Let $(w_1,\ldots,w_p)$ be the vertices of  $C_{k-1}$ from left to right in $\Gamma_{k-1}$, where $w_1=v_1$ and $w_p=v_2$. For the next path $P_k = \{z_1,\ldots,z_q \}$ in $\pi$, let $w_\ell$ and $w_r$ be the leftmost and rightmost neighbors of $P_k$ on $C_{k-1}$ in $\Gamma_{k-1}$, where $1 \leq  \ell < r \leq p$. For the definition of the shift-set $S(v)$ of each vertex $v$ in $P_{k}$, the algorithm identifies two \emph{critical} vertices on the contour $C_{k-1}$, which we denote by $w_{\ell'}$ and $w_{r'}$, such that $\ell <  \ell' \leq r$ and $\ell \leq  r' < r$  (refer to the white-filled vertices of \Cref{fig:chrobak-kant}); note that it is possible to have $w_{\ell'} = w_{r'}$. Vertex $w_{\ell'}$ is the first vertex encountered in the traversal of $C_{k-1}$ starting from $w_{\ell+1}$ towards $w_r$ that either has a neighbor in $P_k$ or the edge $(w_{\ell'}, w_{\ell'+1})$ is blue or black; note that it is possible to have $w_{\ell'} = w_r$. Symmetrically, vertex $w_{r'}$ is the first vertex encountered in the traversal of $C_{k-1}$ starting from $w_{r-1}$ towards $w_\ell$ that either has a neighbor in $P_k$ or the edge $(w_{r'-1},w_{r'})$ is green or black; note that it is possible to have  $w_{r'} = w_\ell$. We refer to $w_{\ell'}$ and $w_{r'}$ as the \emph{left-critical} and \emph{right-critical vertices} of $P_k$. More importantly, since each internal face of $\Gamma_k$ is convex, in the case where $P_k$ is either a chain or a singleton of degree~$2$ in $G_k$, vertices $w_{\ell'}$ and $w_{r'}$ are either consecutive along $C_{k-1}$ or $w_{\ell'}=w_{r'}$ holds.  Once $w_{\ell'}$ and $w_{r'}$ have been identified, the algorithm sets the shift-sets of the vertices $z_1,\ldots,z_q$ of $P_k$ as follows: 
\begin{equation}\label{eq:shift-sets-1}
S(z_1)= \{z_1\} \cup \bigcup_{i=\ell'}^{r'}S(w_i), \qquad \text{and} \qquad S(z_i)= \{z_i\},\text{ for } i=2,\ldots,q
\end{equation}
Furthermore, to guarantee that the resulting drawing is convex, the algorithm updates the shift-sets of $w_\ell$ and $w_r$ of $G_{k-1}$ as follows:
\begin{equation}\label{eq:shift-sets-2}
S(w_\ell)= \bigcup_{i=\ell}^{\ell'-1}S(w_i), \qquad \text{and} \qquad S(w_r)= \bigcup_{i=r'+1}^{r}S(w_i).
\end{equation}
To compute the drawing $\Gamma_k$, the algorithm distinguishes two cases. If $w_\ell$ is saturated in $G_k$ (i.e., $z_1$ is the last neighbor of $w_\ell$ that has not been drawn), then the $x$-coordinate of $z_1$ is the same as the one of $w_\ell$, that is, $x(z_1) = x(w_\ell)$. Otherwise, $x(z_1) = x(w_\ell)+1$. To accommodate the vertices of $P_k$ and to avoid edge-overlaps, the algorithm shifts each vertex in 
\begin{equation}\label{eq:shift-right}
\bigcup_{i=r}^{p} S(w_i).    
\end{equation}
by $q$ units to the right (see \cref{fig:original}). Then, the algorithm places vertex $z_q$ at $(x(z_1)+q-1, y(w_r) + x(w_r) - (x(z_1)+q-1))$, i.e., at the intersection of the line of slope $-1$ through $w_r$ with the vertical line through point $x(z_1)+q-1$. Note that this is a grid point above $w_\ell$ and $ w_r$ due to the contour condition and the shifting of $S(w_r)$. For  $i=1,\ldots,q-1$, vertex $z_i$ of $P_k$ is placed $q-i$ units to the left of $z_q$. Since $(w_\ell,z_1)$ is blue, $(z_q,w_r)$ is green, and the internal edges of $P_k$ (if any) are black, the contour condition of the algorithm is, by construction, maintained after the placement of the vertices of $P_k$ in $\Gamma_k$.

\medskip

\noindent The contour condition together with the shifting procedure described above guarantee \cref{prp:slopes} for the slopes of the edges in $\Gamma_k$.

\begin{figure}[t]
    \centering
    \begin{subfigure}{.48\textwidth}
    \centering
	\includegraphics[page=1]{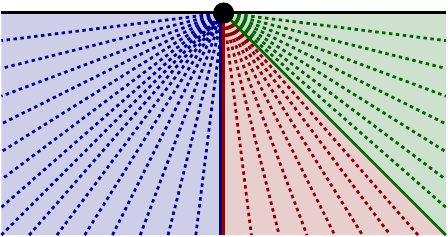}
	\caption{}
	\label{fig:slopes}
	\end{subfigure}
    \begin{subfigure}{.48\textwidth}
    \centering
	\includegraphics[page=1]{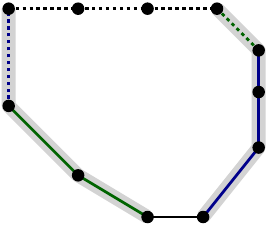}
	\caption{}
	\label{fig:face}
	\end{subfigure}
    \caption{Illustration of (a)~\Cref{prp:slopes} and (b)~\cref{prp:face}.}
\end{figure}

\begin{property}[Chrobak and Kant~\cite{DBLP:journals/ijcga/ChrobakK97}]\label{prp:slopes}
A shift can only decrease the slope of a blue edge, increase the slope of a green edge, while the black and the red edges are \emph{rigid}, i.e., they maintain their slope. As a result, in $\Gamma_k$ (see \cref{fig:slopes}): 
\begin{itemize}[--]\setlength{\itemsep}{0pt}
\item the slope of each blue edge ranges in $(0,+\infty]$,
\item the slope of each black edge is $0$, 
\item the slope of each green edge ranges in $[-1,0)$, and
\item the slope of each red edge ranges in $[-\infty,-1)$.  
\end{itemize}
\end{property}

\noindent Since each face of $\Gamma_k$ is formed when a path $P_{k'}$ with $k' \leq k$ of canonical order $\pi$ is introduced, \Cref{prp:slopes} combined with the contour condition and Property~P.\ref{prp:4} of the canonical order imply the following property for the shape of each face~in~$\Gamma_k$.

\begin{property}\label{prp:face}
Let $f$ be a face in $\Gamma_k$. Then, a counter-clockwise traversal of $f$ starting from its leftmost vertex that is the bottommost when it is not uniquely defined consists of the following boundary parts (see \cref{fig:face}): 
\begin{enumerate}[i.]\setlength{\itemsep}{0pt}
\item\label{f:1}a strictly descendant path of green edges (possibly empty), 
\item\label{f:2}a black edge (possible non existent),
\item\label{f:3}a strictly ascendant path of blue edges (possible empty),
\item\label{f:4}a green or red edge,
\item\label{f:5}a horizontal path of black edges (possibly empty), and 
\item\label{f:6}a blue or red edge.
\end{enumerate}
\end{property}

\noindent Boundary parts~(\ref{f:1})--(\ref{f:3}) in \cref{prp:face} form the \emph{lower envelope} of $f$ (solid in \cref{fig:face}). The \emph{upper envelope} of $f$ is formed by boundary parts~(\ref{f:4})--(\ref{f:6})  (dotted in \cref{fig:face}). The latter is introduced in $\Gamma_k$ when a path of the canonical order is placed. Thus, the upper envelope cannot contain black and red edges simultaneously (by Property~P.\ref{prp:3} of canonical order). Finally, boundary parts~(\ref{f:3}) and~(\ref{f:4}) form the \emph{right envelope} of $f$, while (\ref{f:6}) and~(\ref{f:1}) form the \emph{left envelope}~of~$f$ (gray-highlighted in \cref{fig:face}). We next state some lemmata regarding the ``behavior'' of the algorithm by Chrobak and Kant~\cite{DBLP:journals/ijcga/ChrobakK97} that are employed in the proof of correctness of our modification. 

\begin{lemma}\label{lem:same-y-coordinate}
Let $u$ and $v$ be two distinct vertices of $G_k$ belonging to the same face~$f$ of $\Gamma_k$. If $u$ and $v$ have the same $y$-coordinate in $\Gamma_k$ with $x(u)<x(v)$, then either $u$ and $v$ are connected by a path of black edges of $f$ or the $x$-coordinate of the bottommost vertex/vertices of $f$ is/are in the interval $(x(u),x(v)]$.
\end{lemma}
\begin{proof}
Indeed, if $u$ and $v$ are connected by a path consisting exclusively of black edges, then by \Cref{prp:slopes} $u$ and $v$ have the same $y$-coordinate. Assume now that $u$ and $v$ are not connected by a path of black edges, and let without loss of generality $u$ be to the left of $v$ in $\Gamma_k$. Since the vertices of the left/right envelope all have distinct coordinates, $u$ belongs to the left envelope of $f$, while $v$ to the right envelope of $f$. As a result, a counter-clockwise traversal of $f$ from $u$ to $v$ contains the bottommost vertex/vertices of $f$, which proves the statement. 
\end{proof}

\begin{lemma}\label{lem:martin-property}
Let $\mathcal{S}_k$ be the vertices of $G_{k-1}$ that were shifted during the introduction of $P_k$ in $\Gamma_k$, and let $c$ be a positive integer. Let $\Gamma_k'$ be the drawing obtained from $\Gamma_{k-1}$ by first shifting the vertices of $\mathcal{S}_k$ by $c$ units to the right and then attaching $P_k$ as in the algorithm by Chrobak and Kant~\cite{DBLP:journals/ijcga/ChrobakK97}. Then, $\Gamma_k'$ is a convex planar grid drawing of $G_k$.
\end{lemma}
\begin{proof}
We focus on the case in which $P_k$ is a singleton; the case of a chain is similar. Let $w_{\bar{r}}$ and $w_r$ with $\ell \leq \bar{r} < r$ be the rightmost two neighbors of $P_k$ on $C_{k-1}$ in $\Gamma_{k-1}$. Note that, by construction (see \cref{eq:shift-right}), $\mathcal{S}_k$ is the union of the shift-sets of $w_r$ and all the vertices that follow it in the contour $C_{k-1}$. Then, shifting the vertices of $\mathcal{S}_k$ by $c$ units to the right can be simulated by the following procedure: Attach a degree-$2$ singleton $s_0$ with endpoints $w_r$ and $w_{\bar{r}}$, and for each $i=1,\ldots,c-1$ attach a degree-$3$ singleton $s_i$ with neighbors $w_r$, $s_{i-1}$ and $w_{\bar{r}}$. By the correctness of the algorithm by Chrobak and Kant~\cite{DBLP:journals/ijcga/ChrobakK97}, the resulting drawing is a convex planar grid drawing which satisfies the contour condition. Then, removing vertices $s_0,\ldots,s_{c-1}$ and introducing $P_k$ as in the algorithm does not violate any of the aforementioned properties, yielding a convex planar grid drawing $\Gamma_k'$ of $G_k$.
\end{proof}

\begin{lemma}\label{lem:black-edge}
Let $f$ be a face of $\Gamma_{k-1}$ that contains a black edge $(u,v)$ at its lower envelope, such that $u$ is to the left of $v$ in $\Gamma_{k-1}$, and assume that some vertices of $f$ are not shifted during the introduction of $P_k$ in $\Gamma_k$. Then, neither $u$ nor $v$ are shifted, unless $(u,v)$ is the rightmost edge of the lower envelope of $f$, in which case $u$ is not shifted, while $v$ is shifted. 
\end{lemma}
\begin{proof}
By \cref{eq:shift-sets-2}, it follows that if $v$ is not the rightmost edge of the lower envelope of $f$, then neither vertex $u$ nor vertex $v$ is in the shift-set of the rightmost vertex of $f$, when $f$ is formed in the incremental construction of $\Gamma_{k-1}$ based on $\pi$. On the other hand, if $v$ is the rightmost edge of the lower envelope of face $f$, then $v$ is in the shift-set $S(v)$, while vertex $u$ is not. Since not all vertices of $f$ are shifted during the introduction of $P_k$ in $\Gamma_k$, it follows that, among the vertices of the lower envelope of $f$, the ones that are shifted are those in the shift-set of the rightmost vertex of the lower envelope of $f$, which proves the lemma. 
\end{proof}

\subsection{Our modification.}
\label{ssec:modification}
We start by placing $v_1$ and $v_2$ of path $P_0$ as in the algorithm by Chrobak and Kant~\cite{DBLP:journals/ijcga/ChrobakK97}, that is, at points $(0,0)$ and $(1,0)$, respectively. Assume now that $\Gamma_{k-1}$ is a convex planar \drawing of $G_{k-1}$. For placing path $P_k$ in drawing $\Gamma_{k-1}$,  $k=1,\ldots,m$, we distinguish two cases. In the first case, $P_k$ is a chain and we proceed as in the algorithm by Chrobak and Kant~\cite{DBLP:journals/ijcga/ChrobakK97}. Hence, we focus on the more elaborated case, in which $P_k$ is a singleton, i.e., $P_k=\{z_1\}$. In this case, our algorithm first shifts the vertices of $\Gamma_{k-1}$ appropriately to guarantee that the obtained drawing $\Gamma_k$ is a \drawing (to be shown in~\cref{lem:properties}). Let $w_{x_0},\ldots,w_{x_{\rho+1}}$  be the neighbors of $P_k$ along $C_{k-1}$, such that $\ell = x_0 < x_1 < \ldots < x_{\rho} < x_{\rho+1} = r$. Note that, based on this notation, $\rho$ denotes the number of neighbors of $P_k$ between $w_\ell$ and $w_r$ on $C_{k-1}$. Besides critical vertices $w_{\ell'}$ and $w_{r'}$, our modification introduces the following $\rho+1$ \emph{pivot vertices} $w_{x_1'},\ldots,w_{x_{\rho+1}'}$, where $w_{x_{1}'}=w_{\ell'}$ and $w_{x_{\rho+1}'}=w_{r}$. For $j=2,\ldots,\rho$, the pivot vertex $w_{x_j'}$ (with $x_{j-1} <  x_j' \leq x_j$) is defined as the first vertex encountered in the traversal of $C_{k-1}$ starting from  $w_{x_{j-1}+1}$ towards $w_{x_j}$ that either is neighboring $z_1$ or is followed by an edge of $C_{k-1}$ that is blue or black. In other words, pivot vertex $w_{x_j'}$ would be the vertex that the algorithm by Chrobak and Kant~\cite{DBLP:journals/ijcga/ChrobakK97} identifies as left-critical, when attaching a singleton with exactly two neighbors $w_{x_{j-1}}$ and $w_{x_j}$ on $C_{k-1}$. The algorithm modifies $\Gamma_{k-1}$ by performing $\rho+1$ consecutive refinements of the vertex positions. In the $j$-th refinement, $j=1,\ldots,\rho+1$, the algorithm shifts each vertex in $\bigcup_{i=x_j'}^{p} S(w_i)$ by one unit to the right; see \cref{fig:modification}. This implies that vertices $w_r,\ldots,w_p$ of $C_{k-1}$ have been shifted in total by $\rho+1$ units to the right. Note that in the algorithm by Chrobak and Kant~\cite{DBLP:journals/ijcga/ChrobakK97} these vertices would be shifted by only one unit. The next observations follow from our shifting~strategy.

\begin{observation}\label{obs:chain-shift}
If $P_k$ is a chain, then our shifting strategy and the one by Chrobak and Kant~\cite{DBLP:journals/ijcga/ChrobakK97} are identical.  
\end{observation}

\begin{observation}\label{obs:singleton-shift}
If $P_k$ is a singleton, then the horizontal distance between any two consecutive neighbors of $P_k$ in $C_{k-1}$ gets increased by one unit in $\Gamma_k$, while in the algorithm by Chrobak and Kant~\cite{DBLP:journals/ijcga/ChrobakK97} this would only be the case for $w_{x_\rho}$ and $w_{x_\rho+1}=w_r$. 
\end{observation}

\begin{figure}[t!]
	\centering
    \begin{subfigure}[b]{.4\textwidth}
    \centering
	\includegraphics[page=9,scale=0.9]{figs/canonical}
	\caption{}
	\label{fig:modification}
	\end{subfigure}\hfil
    \begin{subfigure}[b]{.4\textwidth}
    \centering
	\includegraphics[page=1]{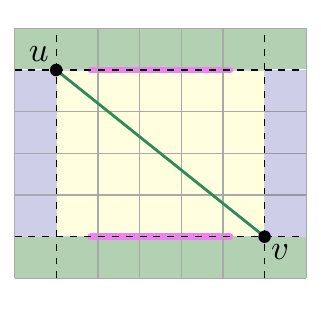}
	\caption{}
	\label{fig:uv}
	\end{subfigure}
	\caption{%
	(a)~Introducing a singleton $P_k$ in $\Gamma_{k-1}$ according to our modification of the algorithm by Chrobak and Kant~\cite{DBLP:journals/ijcga/ChrobakK97}; the white-filled vertices are the identified pivot vertices, and (b)~edge-vertex configurations used in the proof of \Cref{lem:properties}.}	
	
\end{figure}

The construction of $\Gamma_k$ is completed by placing the vertices of $P_k$ as in the algorithm by Chrobak and Kant~\cite{DBLP:journals/ijcga/ChrobakK97}, i.e., we set either $x(z_1) = x(w_\ell)$ or $x(z_1) = x(w_\ell)+1$ (depending on whether $z_1$ is saturated or not, respectively), we place $z_q$ at the intersection of the line of slope~$-1$ through $w_r$ with the vertical line through point $x(z_1)+q-1$ and for $i=1,\ldots,q-1$, vertex $v_i$ of $P_k$ is placed $q-i$ units to the left of $z_q$. Thus, the contour condition is maintained in $\Gamma_k$. 

\begin{lemma}\label{lem:planarity-convexity}
The drawing $\Gamma_k$ produced by our modification of the algorithm by Chrobak and Kant~\cite{DBLP:journals/ijcga/ChrobakK97} is planar and convex. 
\end{lemma}
\begin{proof}
The fact that drawing $\Gamma_k$ is planar is implied by the original proof of Chrobak and Kant~\cite{DBLP:journals/ijcga/ChrobakK97}, since the contour condition is maintained for $\Gamma_k$. We next argue about the convexity of $\Gamma_k$. Since $\Gamma_{k-1}$ is convex, if $P_k$ is a chain, then  $\Gamma_k$ is also convex by \cref{obs:chain-shift}. Assume that $P_k$ is a singleton, i.e., $P_k = \{z_1\}$. In this case, we claim that the extra shifts that our modification performs (see \cref{obs:singleton-shift}) do not affect the convexity of $\Gamma_k$. To prove the claim, consider any two consecutive neighbors $w_{x_j}$ and $w_{x_{j+1}}$ of $z_1$ along $(w_1,\ldots,w_p)$ of $C_{k-1}$ with $0 \leq j \leq \rho$. If the algorithm by Chrobak and Kant were about to place a singleton connecting only $w_{x_j}$ and $w_{x_{j+1}}$ to derive $\Gamma_k$, then it would perform a shift using as left-critical vertex the one that our modification identifies as pivot $w_{x_j'}$, and as right-critical vertex either the same vertex or its right neighbor (since the singleton is of degree $2$). Thus, convexity would be maintained. Applying the same reasoning to any pair of consecutive neighbors of $P_k$, proves that the subdrawing of $\Gamma_k$ induced by $G_{k-1}$ is indeed convex. In addition, the same reasoning implies that \cref{prp:slopes,prp:face} of the algorithm by Chrobak and Kant~\cite{DBLP:journals/ijcga/ChrobakK97} also hold.  To complete the proof of our claim, we note that the fact that the faces incident to $P_k$ in $\Gamma_k$ are  convex follows using the same approach as in the algorithm by Chrobak and Kant, as the contour condition is maintained.  
\end{proof}

\noindent Note that since the contour condition is maintained and we do not modify the shift-sets, the fact that \cref{prp:slopes,prp:face} hold in our modification implies that \cref{lem:same-y-coordinate,lem:martin-property,lem:black-edge} also hold. To complete the proof of correctness of our algorithm, we  prove in the following lemma that $\Gamma_k$ is a \drawing of $G_k$. To ease the proof, we denote by $\Gamma_{k-1}'$ the drawing of $G_{k-1}$ obtained after the preparatory shifting in drawing $\Gamma_{k-1}$ for the introduction of $P_k$.

\begin{lemma}\label{lem:properties}
Let $\Gamma_{k-1}$ be a \drawing of $G_{k-1}$ computed by our algorithm. The following statements hold: 
\begin{inparaenum}[(i)]
\item\label{part1}the edge-vertex resolution of $\Gamma_{k-1}'$ is no less than that of $\Gamma_{k-1}$, and
\item\label{part2}introducing the new edges of $\Gamma_k$, which are either part of $P_k$ or incident to the endpoints of $P_k$, preserves the edge-vertex resolution to at least $\frac{1}{2}$.
\end{inparaenum}
\end{lemma}
\begin{proof}
Since the drawing of $G_{k-1}$ is planar in $\Gamma_{k}$, it is sufficient to only consider its faces in order to prove statement~(\ref{part1}). To this end, consider any arbitrary face $f$ in $\Gamma_{k-1}'$. If either none or all of the vertices of $f$ are shifted by the same amount, then statement~(\ref{part1}) obviously holds. Consider now the case where $f$ contains at least one vertex that is shifted and one vertex that remains stationary in $\Gamma_{k-1}'$. Suppose, for a contradiction, that $f$ contains an edge $(u,v)$ and a vertex $w$ that is not incident to $(u,v)$ such that $(u,v)$ intersects the disk of $w$ in $\Gamma_{k-1}'$. Since $\Gamma_{k-1}'$ is a grid drawing of $G_{k-1}$, it follows by \cref{prp:slopes} that $(u,v)$ cannot be a black edge. Assume that the slope of $(u,v)$ is negative (i.e., $(u,v)$ is green or red), as the case in which it is positive (i.e., $(u,v)$ is blue) is similar. W.l.o.g., further assume that $u$ is above $v$ in $\Gamma_{k-1}'$ (see \cref{fig:uv}).

Since $\Gamma_{k-1}'$ is a grid drawing of $G_{k-1}$, it follows that, regardless of whether $w$ was shifted or not, $w$ is neither above nor below the \emph{horizontal} strip delimited by the two horizontal lines through $u$ and $v$ in $\Gamma_{k-1}'$ (green in \cref{fig:uv}). Similarly, one observes that $w$ is neither to the left nor to the right of the \emph{vertical} strip delimited by the two vertical lines through $u$ and $v$ (blue in \cref{fig:uv}). It follows that $w$ is either in the interior (yellow in \cref{fig:uv}) or on the boundary of the axis-aligned bounding box $B_{uv}$ of the edge $(u,v)$ in $\Gamma_{k-1}'$. 

We next argue that $w$ can be neither in the interior of  $B_{uv}$ nor along its two vertical sides, which implies that $w$ is necessarily on one of the two horizontal sides of $B_{uv}$ (purple in \cref{fig:uv}). To see this, assume for a contradiction that $w$ is in the interior of  $B_{uv}$ or along one of its two vertical sides but not at its corners. Since $\Gamma_{k-1}$ is a \drawing of $G_{k-1}$ while $\Gamma_{k-1}'$ is not, it follows that the distance between $w$ and $(u,v)$ decreased after the shifting (by one unit) to obtain $\Gamma_{k-1}'$ from $\Gamma_{k-1}$.  If the shifting were sufficiently large (and greater than one unit), then $w$ would be on different sides of $(u,v)$ in $\Gamma_{k-1}$ and in $\Gamma_{k-1}'$, violating the planarity of the drawing (which is implied by \cref{lem:martin-property}); a contradiction.

It follows that $w$ is on any of the two horizontal sides of $B_{uv}$, as we initially claimed. We proceed by considering two subcases depending on whether $(u,v)$ is on the upper or lower envelope of $f$. Consider first the case where $(u,v)$ is on the upper envelope of $f$. By \cref{prp:face}, it follows that $w$ is on the lower envelope of $f$ and, thus, on the lower edge of $B_{uv}$. If $w$ and $v$ are not adjacent, the fact that $w$ and $v$ have the same $y$-coordinate implies that the $x$-coordinate of the bottomost vertex/vertices of $f$ is delimited by $w$ and $v$ (by \cref{lem:same-y-coordinate}). This further implies that $w$ and $v$ are on the left and right envelopes of $f$, respectively. Since not all vertices of $f$ are shifted  in $\Gamma_{k-1}'$, it follows that, among the vertices of the lower envelope of $f$, the ones that are shifted are those in the shift-set of the rightmost vertex of the lower envelope of $f$, which implies that $w$ has not been shifted. On the other hand, if $w$ and $v$ are adjacent, then the edge connecting them is black (by \cref{prp:slopes}), and thus by \cref{lem:black-edge}, we conclude again that $w$ is not shifted. In both cases, however, the edge-vertex resolution of $\Gamma_{k-1}'$ cannot be smaller than the one of $\Gamma_{k-1}$; a contradiction. 

Consider now the case where $(u,v)$ is on the lower envelope of $f$. In this case, vertex $w$ can be either on the lower or on the upper envelope of $f$. The former case can be ruled out by adopting an argument similar to the one of the previous paragraph. In the latter case, vertex $w$ is on the top side of $B_{uv}$. By \cref{prp:face}, vertices $u$ and $w$ are connected by a path of black edges contradicting the fact that $(u,v)$ is green, since, by \cref{prp:face}.\ref{f:6}, the left edge connecting to a black path on the upper envelope~is~blue.

We now prove statement~(\ref{part2}). Assume for a contradiction that an edge $(u,v)$ added in $\Gamma_k$ during the introduction of $P_k$ intersects the disk of a vertex $w$. Clearly, $w$ belongs to $G_{k-1}$, since by construction we have no edge-disk intersections between elements of $P_k$. Thus, one endpoint of $(u,v)$, say $u$, belongs to $P_k$ and the other, say $v$, to $G_{k-1}$, i.e., $(u,v)$ is not a black edge. If $(u,v)$ is green, then its slope in $\Gamma_k$ is $-1$, and hence it cannot intersect any non-adjacent disk. Assume first that $(u,v)$ is blue, and observe that the construction is such that the horizontal distance between $u$ and $v$ is either $0$ or $1$. In the former case, $(u,v)$ is vertical and cannot intersect any non-adjacent vertex-disk. In the latter case, the shifting performed by the algorithm guarantees that the part of the grid column along which vertex $u$ is placed that is contained in the horizontal strip bounded by the horizontal lines through $u$ and $v$ in $\Gamma_k$ contains no vertex of $C_{k-1}$, and therefore edge $(u,v)$ again cannot intersect any non-adjacent vertex disk. The argument for the case in which $(u,v)$ is red is analogous. 
\end{proof}

\noindent We conclude the proof of our main result by analyzing the area of the produced drawings and the time complexity of the algorithm. 

\myparagraph{Proof of \cref*{th:main}.} Let $G$ be an $n$-vertex $3$-connected plane graph with $f$ internal faces. Let $\Gamma$ be a planar drawing of $G$ computed by our algorithm. By \cref{lem:planarity-convexity}, drawing $\Gamma$ is convex, and, by \cref{lem:properties}, its edge-vertex resolution is at least $\frac{1}{2}$. By the contour condition, $\Gamma$ is inside a right isosceles triangle, such that it has a horizontal side (which corresponds to edge $(v_1,v_2)$) and a vertical side (which contains edge $(v_1,v_n)$) that have the same length and meet at point $(0,0)$. In the algorithm by Chrobak and Kant, the value of the width and the height of this triangle is $n-2$~\cite{DBLP:journals/ijcga/ChrobakK97}. The additional unit-shifts due to the introduction of singletons performed by our modification  increase the value of the width and the height by the same amount $a$ (see \cref{obs:chain-shift,obs:singleton-shift}).  We focus on the width of $\Gamma$, and we distinguish two cases: either $\min\{f,n-3\} = n-3$ or $\min\{f,n-3\} = f$. 

\begin{itemize}[--]
\item Assume first that $\min\{f,n-3\} = n-3$. We develop a charging argument that charges each additional one-unit shift to the red edges of $G$. In particular, consider a singleton $P_k$. The additional shifts due to this singleton are two less than its degree in $G_{k}$, which equals the number of red edges incident to $P_k$ in $G_k$. It is immediate to see that each red edge is charged to exactly one additional shift. Hence, the total number of additional shifts is at most the number of red edges in $G$, which is at most $n-3$ (recall that the red subgraph of $G$ is a forest with at most  $n-3$ edges). Consequently, in this case $a \le n-3$.
    
\item Assume now that $\min\{f,n-3\} = f$. In this case we develop a similar charging argument, in which we charge each additional one-unit shift to the internal faces of $G$, rather than to its red edges. Again, consider a singleton $P_k$, and observe that the additional shifts due to this singleton are two less than its degree in $G_{k}$. This value equals the number of internal faces incident to $P_k$ in $G_k$ minus one, in particular, we can avoid charging the shift to the rightmost internal face incident to $P_k$. It is not difficult to see that each internal face is charged to at most one additional shift. Hence, the total number of additional shifts is at most the number of internal faces $f$ in $G$. Consequently, in this case $a \le f$.
\end{itemize}

Finally, we discuss the time complexity. The algorithm by Chrobak and Kant can be implemented to run in linear time~\cite{DBLP:journals/ijcga/ChrobakK97}. In particular, the key ingredient to achieve linear time complexity, is the use of relative coordinates for the vertices, which avoids shifting entire subgraphs. Since our algorithm only requires a linear number of additional one-unit shifts and it does not modify the shift-sets of the vertices, this translates into different relative coordinates and requires neither additional operations nor different data structures. Therefore it can be implemented to also run in linear time. 
\qed

\section{Open Problems}
\label{sec:conclusions}
In this work, we present improvements upon results in~\cite{DBLP:journals/comgeo/BekosGMPST21,DBLP:conf/compgeom/ChrobakGT96}. The following research directions naturally stem from our work.
\begin{enumerate}[(i)]
    \item Can the bounded edge-vertex resolution requirement be incorporated into an area lower bound so to improve the one given in \cite{DBLP:conf/stoc/FraysseixPP88,DBLP:journals/algorithmica/Kant96}?
    \item Can the area bound of \cref{th:main} be improved in the case where the input graph is 4-connected? Note that such graphs admit  $W \times H$ drawings with $W + H \leq n-1$~\cite{DBLP:journals/ijfcs/MiuraNN06} but their edge-vertex resolution may be arbitrarily small.
    \item Finally, it is of interest to study the edge-vertex resolution requirement for strictly convex drawings.
\end{enumerate}

\bibliographystyle{abbrvurl}
\bibliography{references}

\begin{thebibliography}{10}

\bibitem{DBLP:conf/gd/AlamdariB12}
S.~Alamdari and T.~C. Biedl.
\newblock Open rectangle-of-influence drawings of non-triangulated planar
  graphs.
\newblock In W.~Didimo and M.~Patrignani, editors, {\em {GD}}, volume 7704 of
  {\em LNCS}, pages 102--113. Springer, 2012.
\newblock \href {http://dx.doi.org/10.1007/978-3-642-36763-2\_10}
  {\path{doi:10.1007/978-3-642-36763-2\_10}}.

\bibitem{Barany2006}
I.~B{\'{a}}r{\'{a}}ny and G.~Rote.
\newblock Strictly convex drawings of planar graphs.
\newblock {\em Documenta Mathematica}, 11:369--391, 2006.
\newblock URL: \url{http://eudml.org/doc/53043}.

\bibitem{DBLP:journals/jgaa/BarequetGR04}
G.~Barequet, M.~T. Goodrich, and C.~Riley.
\newblock Drawing planar graphs with large vertices and thick edges.
\newblock {\em J. Graph Algorithms Appl.}, 8:3--20, 2004.
\newblock \href {http://dx.doi.org/10.7155/jgaa.00078}
  {\path{doi:10.7155/jgaa.00078}}.

\bibitem{DBLP:journals/dm/BarriereH12}
L.~Barri{\`{e}}re and C.~Huemer.
\newblock 4-labelings and grid embeddings of plane quadrangulations.
\newblock {\em Discret. Math.}, 312(10):1722--1731, 2012.
\newblock \href {http://dx.doi.org/10.1016/j.disc.2012.01.027}
  {\path{doi:10.1016/j.disc.2012.01.027}}.

\bibitem{DBLP:journals/comgeo/BekosGMPST21}
M.~A. Bekos, M.~Gronemann, F.~Montecchiani, D.~P{\'{a}}lv{\"{o}}lgyi,
  A.~Symvonis, and L.~Theocharous.
\newblock Grid drawings of graphs with constant edge-vertex resolution.
\newblock {\em Comput. Geom.}, 98:101789, 2021.
\newblock \href {http://dx.doi.org/10.1016/j.comgeo.2021.101789}
  {\path{doi:10.1016/j.comgeo.2021.101789}}.

\bibitem{DBLP:conf/gd/BiedlBM99}
T.~C. Biedl, A.~Bretscher, and H.~Meijer.
\newblock Rectangle of influence drawings of graphs without filled 3-cycles.
\newblock In J.~Kratochv{\'{\i}}l, editor, {\em {GD}}, volume 1731 of {\em
  LNCS}, pages 359--368. Springer, 1999.
\newblock \href {http://dx.doi.org/10.1007/3-540-46648-7\_37}
  {\path{doi:10.1007/3-540-46648-7\_37}}.

\bibitem{DBLP:conf/cccg/BiedlLMV16}
T.~C. Biedl, A.~Lubiw, S.~Mehrabi, and S.~Verdonschot.
\newblock Rectangle-of-influence triangulations.
\newblock In T.~C. Shermer, editor, {\em {CCCG}}, pages 237--243, 2016.

\bibitem{DBLP:journals/algorithmica/BonichonFM07}
N.~Bonichon, S.~Felsner, and M.~Mosbah.
\newblock Convex drawings of 3-connected plane graphs.
\newblock {\em Algorithmica}, 47(4):399--420, 2007.
\newblock URL: \url{https://doi.org/10.1007/s00453-006-0177-6}, \href
  {http://dx.doi.org/10.1007/s00453-006-0177-6}
  {\path{doi:10.1007/s00453-006-0177-6}}.

\bibitem{CON85}
N.~Chiba, K.~Onoguchi, and T.~Nishizeki.
\newblock Drawing planar graphs nicely.
\newblock {\em Acta Inform.}, 22:187–201, 1985.
\newblock \href {http://dx.doi.org/10.1007/BF00264230}
  {\path{doi:10.1007/BF00264230}}.

\bibitem{DBLP:conf/compgeom/ChrobakGT96}
M.~Chrobak, M.~T. Goodrich, and R.~Tamassia.
\newblock Convex drawings of graphs in two and three dimensions (preliminary
  version).
\newblock In S.~Whitesides, editor, {\em {SoCG}}, pages 319--328. {ACM}, 1996.
\newblock \href {http://dx.doi.org/10.1145/237218.237401}
  {\path{doi:10.1145/237218.237401}}.

\bibitem{DBLP:journals/ijcga/ChrobakK97}
M.~Chrobak and G.~Kant.
\newblock Convex grid drawings of 3-connected planar graphs.
\newblock {\em Int. J. Comput. Geom. Appl.}, 7(3):211--223, 1997.
\newblock \href {http://dx.doi.org/10.1142/S0218195997000144}
  {\path{doi:10.1142/S0218195997000144}}.

\bibitem{DBLP:journals/ipl/ChrobakP95}
M.~Chrobak and T.~H. Payne.
\newblock A linear-time algorithm for drawing a planar graph on a grid.
\newblock {\em Inf. Process. Lett.}, 54(4):241--246, 1995.
\newblock \href {http://dx.doi.org/10.1016/0020-0190(95)00020-D}
  {\path{doi:10.1016/0020-0190(95)00020-D}}.

\bibitem{DBLP:conf/stoc/FraysseixPP88}
H.~de~Fraysseix, J.~Pach, and R.~Pollack.
\newblock Small sets supporting {F\'{a}}ry embeddings of planar graphs.
\newblock In J.~Simon, editor, {\em {STOC}}, pages 426--433. {ACM}, 1988.
\newblock \href {http://dx.doi.org/10.1145/62212.62254}
  {\path{doi:10.1145/62212.62254}}.

\bibitem{Battista2013}
G.~{Di Battista} and F.~Frati.
\newblock Drawing trees, outerplanar graphs, series-parallel graphs, and planar
  graphs in a small area.
\newblock In J.~Pach, editor, {\em Thirty Essays on Geometric Graph Theory},
  pages 121--165. Springer New York, 2013.
\newblock URL: \url{https://doi.org/10.1007/978-1-4614-0110-0_9}, \href
  {http://dx.doi.org/10.1007/978-1-4614-0110-0_9}
  {\path{doi:10.1007/978-1-4614-0110-0_9}}.

\bibitem{DBLP:journals/corr/BattistaF14}
G.~{Di Battista} and F.~Frati.
\newblock A survey on small-area planar graph drawing.
\newblock {\em CoRR}, abs/1410.1006, 2014.

\bibitem{DBLP:journals/algorithmica/BattistaTV99}
G.~{Di Battista}, R.~Tamassia, and L.~Vismara.
\newblock Output-sensitive reporting of disjoint paths.
\newblock {\em Algorithmica}, 23(4):302--340, 1999.
\newblock URL: \url{https://doi.org/10.1007/PL00009264}, \href
  {http://dx.doi.org/10.1007/PL00009264} {\path{doi:10.1007/PL00009264}}.

\bibitem{DBLP:journals/ijfcs/DuncanEKW06}
C.~A. Duncan, A.~Efrat, S.~G. Kobourov, and C.~Wenk.
\newblock Drawing with fat edges.
\newblock {\em Int. J. Found. Comput. Sci.}, 17(5):1143--1164, 2006.
\newblock \href {http://dx.doi.org/10.1142/S0129054106004315}
  {\path{doi:10.1142/S0129054106004315}}.

\bibitem{DBLP:journals/order/Felsner01}
S.~Felsner.
\newblock Convex drawings of planar graphs and the order dimension of
  3-polytopes.
\newblock {\em Order}, 18(1):19--37, 2001.
\newblock URL: \url{https://doi.org/10.1023/A:1010604726900}, \href
  {http://dx.doi.org/10.1023/A:1010604726900}
  {\path{doi:10.1023/A:1010604726900}}.

\bibitem{felsner}
S.~Felsner.
\newblock {\em Geometric Graphs and Arrangements}.
\newblock Advanced Lectures in Mathematics. Vieweg, 2004.
\newblock \href {http://dx.doi.org/10.1007/978-3-322-80303-0}
  {\path{doi:10.1007/978-3-322-80303-0}}.

\bibitem{Far48}
I.~Fáry.
\newblock On straight lines representation of planar graphs.
\newblock {\em Acta Sci. Math. (Szeged)}, 11:229–233, 1948.

\bibitem{DBLP:journals/dcg/He97}
X.~He.
\newblock Grid embedding of 4-connected plane graphs.
\newblock {\em Discret. Comput. Geom.}, 17(3):339--358, 1997.
\newblock URL: \url{https://doi.org/10.1007/PL00009290}, \href
  {http://dx.doi.org/10.1007/PL00009290} {\path{doi:10.1007/PL00009290}}.

\bibitem{DBLP:journals/algorithmica/Kant96}
G.~Kant.
\newblock Drawing planar graphs using the canonical ordering.
\newblock {\em Algorithmica}, 16(1):4--32, 1996.
\newblock \href {http://dx.doi.org/10.1007/BF02086606}
  {\path{doi:10.1007/BF02086606}}.

\bibitem{DBLP:journals/dcg/MiuraMN09}
K.~Miura, T.~Matsuno, and T.~Nishizeki.
\newblock Open rectangle-of-influence drawings of inner triangulated plane
  graphs.
\newblock {\em Discret. Comput. Geom.}, 41(4):643--670, 2009.
\newblock \href {http://dx.doi.org/10.1007/s00454-008-9098-2}
  {\path{doi:10.1007/s00454-008-9098-2}}.

\bibitem{DBLP:journals/dcg/MiuraNN01}
K.~Miura, S.~Nakano, and T.~Nishizeki.
\newblock Grid drawings of 4-connected plane graphs.
\newblock {\em Discret. Comput. Geom.}, 26(1):73--87, 2001.
\newblock URL: \url{https://doi.org/10.1007/s00454-001-0004-4}, \href
  {http://dx.doi.org/10.1007/s00454-001-0004-4}
  {\path{doi:10.1007/s00454-001-0004-4}}.

\bibitem{DBLP:journals/ijfcs/MiuraNN06}
K.~Miura, S.~Nakano, and T.~Nishizeki.
\newblock Convex grid drawings of four-connected plane graphs.
\newblock {\em Int. J. Found. Comput. Sci.}, 17(5):1031--1060, 2006.
\newblock URL: \url{https://doi.org/10.1142/S012905410600425X}, \href
  {http://dx.doi.org/10.1142/S012905410600425X}
  {\path{doi:10.1142/S012905410600425X}}.

\bibitem{DBLP:books/ws/NishizekiR04}
T.~Nishizeki and M.~S. Rahman.
\newblock {\em Planar Graph Drawing}, volume~12 of {\em Lecture Notes Series on
  Computing}.
\newblock World Scientific, 2004.
\newblock \href {http://dx.doi.org/10.1142/5648} {\path{doi:10.1142/5648}}.

\bibitem{DBLP:conf/soda/Rote05}
G.~Rote.
\newblock Strictly convex drawings of planar graphs.
\newblock In {\em {SODA}}, pages 728--734. {SIAM}, 2005.
\newblock URL: \url{http://dl.acm.org/citation.cfm?id=1070432.1070535}.

\bibitem{DBLP:journals/comgeo/SadasivamZ11}
S.~Sadasivam and H.~Zhang.
\newblock Closed rectangle-of-influence drawings for irreducible
  triangulations.
\newblock {\em Comput. Geom.}, 44(1):9--19, 2011.
\newblock \href {http://dx.doi.org/10.1016/j.comgeo.2010.07.001}
  {\path{doi:10.1016/j.comgeo.2010.07.001}}.

\bibitem{DBLP:conf/soda/Schnyder90}
W.~Schnyder.
\newblock Embedding planar graphs on the grid.
\newblock In {\em {SODA}}, pages 138--148. {SIAM}, 1990.
\newblock URL: \url{http://dl.acm.org/citation.cfm?id=320176.320191}.

\bibitem{St51}
S.~K. Stein.
\newblock Convex maps.
\newblock {\em Proc. American Math. Soc.}, 2(3):464–466, 1951.

\bibitem{SH34}
E.~Steinitz and H.~Rademacher.
\newblock {\em Vorlesungen {\"u}ber die Theorie der Polyeder}.
\newblock Julius Springer, Berlin, Germany, 1934.

\bibitem{DBLP:journals/jct/Thomassen84}
C.~Thomassen.
\newblock A refinement of {K}uratowski's theorem.
\newblock {\em J. Comb. Theory, Ser. {B}}, 37(3):245--253, 1984.
\newblock URL: \url{https://doi.org/10.1016/0095-8956(84)90057-1}, \href
  {http://dx.doi.org/10.1016/0095-8956(84)90057-1}
  {\path{doi:10.1016/0095-8956(84)90057-1}}.

\bibitem{Tu63}
W.~T. Tutte.
\newblock How to draw a graph.
\newblock {\em Proc. London Math. Soc.}, 13:743--768, 1963.

\bibitem{DBLP:journals/comgeo/Kreveld11}
M.~J. {van Kreveld}.
\newblock Bold graph drawings.
\newblock {\em Comput. Geom.}, 44(9):499--506, 2011.
\newblock \href {http://dx.doi.org/10.1016/j.comgeo.2011.06.002}
  {\path{doi:10.1016/j.comgeo.2011.06.002}}.

\bibitem{DBLP:reference/crc/Vismara13}
L.~Vismara.
\newblock Planar straight-line drawing algorithms.
\newblock In R.~Tamassia, editor, {\em Handbook on Graph Drawing and
  Visualization}, pages 193--222. Chapman and Hall/CRC, 2013.

\bibitem{Wag36}
K.~Wagner.
\newblock Bemerkungen zum {V}ierfarbenproblem.
\newblock {\em Jahresbericht der Deutschen Mathematiker-Vereinigung},
  46:26–32, 1936.

\bibitem{DBLP:conf/wads/ZhangH03}
H.~Zhang and X.~He.
\newblock Compact visibility representation and straight-line grid embedding of
  plane graphs.
\newblock In F.~K. H.~A. Dehne, J.~Sack, and M.~H.~M. Smid, editors, {\em
  {WADS}}, volume 2748 of {\em LNCS}, pages 493--504. Springer, 2003.
\newblock URL: \url{https://doi.org/10.1007/978-3-540-45078-8\_43}, \href
  {http://dx.doi.org/10.1007/978-3-540-45078-8\_43}
  {\path{doi:10.1007/978-3-540-45078-8\_43}}.

\end{thebibliography}
\end{document}